\documentclass[sigconf, nonacm]{acmart}

\newcommand\vldbdoi{XX.XX/XXX.XX}
\newcommand\vldbpages{XXX-XXX}
\newcommand\vldbvolume{19}
\newcommand\vldbissue{1}
\newcommand\vldbyear{2026}
\newcommand\vldbauthors{\authors}
\newcommand\vldbtitle{\shorttitle}
\newcommand\vldbavailabilityurl{https://github.com/mpi-dsg/grappa}
\newcommand\vldbpagestyle{empty}

\makeatletter
\def\input@path{{metadata/}}
\makeatother

\usepackage[a-2b]{pdfx}
\usepackage{enumitem}
\usepackage{subcaption}
\usepackage{multirow}
\usepackage{soul}
\usepackage[linesnumbered,lined,ruled,vlined]{algorithm2e}
\usepackage{balance}

\newcommand{\systemname}{Grappa}
\newtheorem*{theorem*}{Theorem}

\soulregister{\systemname}{7} %

\begin{document}
\title{\systemname{}: Gradient-Only Communication for Scalable Graph Neural Network Training}

\author{Chongyang Xu}
\affiliation{%
 \institution{Max Planck Institute for Software Systems (MPI-SWS)}
 \city{Saarbrücken}
 \country{Germany}
}
\email{cxu@mpi-sws.org}

\author{Christoph Siebenbrunner}
\affiliation{%
 \institution{Vienna University of Economics and Business (WU)}
 \city{Vienna}
 \country{Austria}
}
\email{christoph.siebenbrunner@wu.ac.at}

\author{Laurent Bindschaedler}
\affiliation{%
 \institution{Max Planck Institute for Software Systems (MPI-SWS)}
 \city{Saarbrücken}
 \country{Germany}
}
\email{bindsch@mpi-sws.org}

\renewcommand{\shorttitle}{\systemname{}: Gradient-Only Communication for Scalable Graph Neural Network Training}

\begin{abstract}
Cross‑partition edges dominate the cost of distributed GNN training: fetching remote features and activations per iteration overwhelms the network as graphs deepen and partition counts grow. \systemname{} is a distributed GNN training framework that enforces \emph{gradient‑only communication}: during each iteration, partitions train in isolation and exchange only gradients for the global update. To recover accuracy lost to isolation, \systemname{} (i) periodically \emph{repartitions} to expose new neighborhoods and (ii) applies a lightweight \emph{coverage‑corrected gradient aggregation} inspired by importance sampling. We present an asymptotically unbiased estimator for gradient correction, which we use to develop a minimum-distance batch-level variant that is compatible with common deep-learning packages. %
We also introduce a shrinkage version that improves stability in practice. Empirical results on real and synthetic graphs show that \systemname{} trains GNNs $4 \times$ faster on average (up to $13 \times$) than state-of-the-art systems, achieves better accuracy especially for deeper models, and sustains training at the trillion‑edge scale on commodity hardware. \systemname{} is model‑agnostic, supports full‑graph and mini‑batch training, and does not rely on high‑bandwidth interconnects or caching.
\end{abstract}

\maketitle

\pagestyle{\vldbpagestyle}
\begingroup\small\noindent\raggedright\textbf{PVLDB Reference Format:}\\
\vldbauthors. \vldbtitle. PVLDB, \vldbvolume(\vldbissue): \vldbpages, \vldbyear.\\
\href{https://doi.org/\vldbdoi}{doi:\vldbdoi}
\endgroup
\begingroup
\renewcommand\thefootnote{}\footnote{\noindent
This work is licensed under the Creative Commons BY-NC-ND 4.0 International License. Visit \url{https://creativecommons.org/licenses/by-nc-nd/4.0/} to view a copy of this license. For any use beyond those covered by this license, obtain permission by emailing \href{mailto:info@vldb.org}{info@vldb.org}. Copyright is held by the owner/author(s). Publication rights licensed to the VLDB Endowment. \\
\raggedright Proceedings of the VLDB Endowment, Vol. \vldbvolume, No. \vldbissue\ %
ISSN 2150-8097. \\
\href{https://doi.org/\vldbdoi}{doi:\vldbdoi} \\
}\addtocounter{footnote}{-1}\endgroup

\ifdefempty{\vldbavailabilityurl}{}{
\vspace{.3cm}
\begingroup\small\noindent\raggedright\textbf{PVLDB Artifact Availability:}\\
The source code, data, and/or other artifacts have been made available at \url{\vldbavailabilityurl}.
\endgroup
}

\section{Introduction}
\label{sec:intro}

Graph Neural Networks (GNNs) are a widespread class of machine learning models that learn representations from graph-structured data. GNNs power applications in science, social networks, commerce, and finance~\cite{drug-discovery,protein-folding,friends,community,recommendation,supply-chain,fraud-detection}. However, as graphs and models grow, distributed training encounters a network communication bottleneck, creating a need for efficient, scalable, and distributed GNN training. This paper addresses that need with a training framework that drastically reduces the communication bottleneck and is \emph{model-agnostic}, \emph{theory-backed}, and \emph{empirically validated}.

Distributed GNN training partitions the sparse graph and the dense per-node feature vectors to fit memory and exploit parallelism. Training is iterative: each iteration runs a $k$-layer forward pass that performs message passing, where each node aggregates and transforms neighbor features into hidden activations, followed by a backward pass that computes gradients, and finally a synchronized update aggregates them across replicas. Partitioning breaks edges, so when the forward pass follows a cross-partition edge, the trainer must fetch the remote endpoint's feature at the first layer and the neighbor's activation at deeper layers. As partition counts or model depth increase, these per-iteration remote fetches become the dominant cost, saturate the network, and throttle scalability~\cite{DBLP:conf/osdi/GandhiI21,DBLP:journals/pvldb/ZhangHL0HSGWZ020,bgl}. Mini-batch sampling~\cite{fastgcn,case-for-sampling} reduces computation, but k-hop neighborhoods still span partitions, so remote feature/activation reads remain; sampling reduces volume, not the need for cross-partition fetches. High-speed interconnects, such as RDMA or NVLink, can help; however, they are costly to deploy at scale and can introduce imbalance~\cite{DBLP:conf/usenix/SunSSSWWZLYZW23,dgcl,DBLP:conf/nsdi/LiuC00ZHPCCG23}. Caching reduces traffic by reusing fetched data, but it increases memory pressure and complexity~\cite{DBLP:conf/osdi/GandhiI21,DBLP:journals/pvldb/ZhangHL0HSGWZ020,ugache}. Therefore, current strategies to address this performance bottleneck remain inadequate.

This paper introduces \systemname{}\footnote{\textbf{Grap}h \textbf{pa}rtitioning}, a model-agnostic distributed framework for scalable GNN training on large graphs. \systemname{} \emph{eliminates cross-partition neighbor traffic during iterations}: partitions train in isolation and exchange only gradients on the forward/backward critical path (\emph{gradient-only communication}). Feature movement is deferred to periodic repartitioning boundaries, where it can be amortized. This approach reduces communication and synchronization overhead without specialized interconnects or caching-induced memory pressure. Moreover, partition independence within an iteration allows phase-parallel execution, which trades time for capacity by training a subset of partitions per phase.

While this strategy minimizes communication costs, it introduces a new challenge: a potential decrease in model accuracy due to restricted sample diversity, as each partition only sees its local information. To recover accuracy, \systemname{} combines (1) \emph{dynamic repartitioning} across groups of epochs (super-epochs) to expose new neighborhoods and (2) \emph{coverage-corrected gradient aggregation}, a batch-level importance weighting that compensates for isolation-induced sampling bias. Intuitively, repartitioning is more cost-effective than on-the-fly neighbor fetches because the input graph is typically far smaller than the volume of per-iteration intermediate activations that would otherwise be exchanged.

\textbf{Formal guarantees.} We derive a batch-level corrected gradient estimator and prove it is \emph{asymptotically unbiased} under standard support and boundedness conditions. We also show that, among all ways to rescale a batch with a single number, our choice makes the corrected batch gradient as close as possible (on average) to the ideal importance-weighted gradient (i.e., it minimizes average squared error).%

We have implemented \systemname{} and evaluated its performance in our cluster. Across real-world and RMAT graph datasets, \systemname{} achieves $4\times$ average speedup (up to $13\times$) over state‑of‑the‑art systems, improves test accuracy for deeper models, scales near‑linearly to 64 partitions, and trains a trillion-edge graph in phase-parallel mode on a single commodity server.

\textbf{Terminology.} We use \emph{gradient-only communication} to mean that there is no cross-partition exchange of features or activations during forward or backward passes; only gradients are synchronized across participating replicas.

The paper makes the following contributions:
\begin{itemize}
 \item \emph{Gradient-only communication:} a data-parallel scheme where iterations involve no cross-partition feature/activation traffic; only gradients are synchronized. Halo features move at repartitioning boundaries.
 \item \emph{Theory-backed accuracy recovery:} dynamic repartitioning plus batch-level coverage-corrected aggregation. We prove an asymptotically unbiased node-level estimator and show that the batch-level variant minimizes the mean squared deviation from the ideal.
 \item \emph{Flexible execution:} phase-parallel execution mode that trades time for capacity, allowing single-machine training at the trillion-edge scale.
 \item \emph{Implementation and evaluation:} \systemname{} supports full-graph and mini-batch training and shows superior performance, accuracy, and scalability over current systems.
\end{itemize}

\section{Background \& Motivation}
\label{sec:background}

\paragraph{GNNs 101}
A Graph Neural Network (GNN) takes a sparse graph (vertices, edges) with dense node features and learns node (and edge/graph) representations that encode both structure and features for downstream tasks. A \(k\)-layer forward pass performs message passing, so each node aggregates and transforms information from its neighbors and reaches \(k\) hops. Two execution modes are common: \emph{full-graph} (all nodes and edges participate each iteration) and \emph{mini-batch sampling} (sample target nodes and their multi-hop neighborhoods). Training is \emph{iterative}: every iteration runs a forward pass, computes gradients via backpropagation, and updates parameters with gradient descent.

\paragraph{Requirements and Goals}
Efficient training on large graphs requires a \emph{distributed} system that provides sufficient \emph{capacity} (graph, features, and intermediate states in aggregate memory) and that \emph{parallelizes} computation across partitions to reduce training time. Given the compute intensity of deep models, \emph{accelerators} (GPUs) are essential. Our goal is a \emph{model-agnostic} system that is not tied to a specific GNN architecture, and that targets \emph{ever-larger} graph datasets. We summarize these requirements to motivate why we target gradient-only communication rather than more aggressive partitioning or caching.

\paragraph{Capacity Limits}
Pure GPU full-graph training is often infeasible at scale due to limited High-Bandwidth Memory (HBM). For example, an H100 has 80\,GB, and even an 8-GPU node offers $<$640\,GB for \emph{all} topology, features, parameters, and activations~\cite{DBLP:conf/sigmod/WangZWCZY22,DBLP:conf/mlsys/WanLLKL22}. Several common datasets exceed this budget~\cite{igb-260m,rmat}. Mini-batch sampling can offload graphs to CPU memory, but this merely shifts the capacity bottleneck to DRAM and I/O contention.

\begin{figure}[t]
 \raggedleft
 \includegraphics[width=0.47\textwidth]{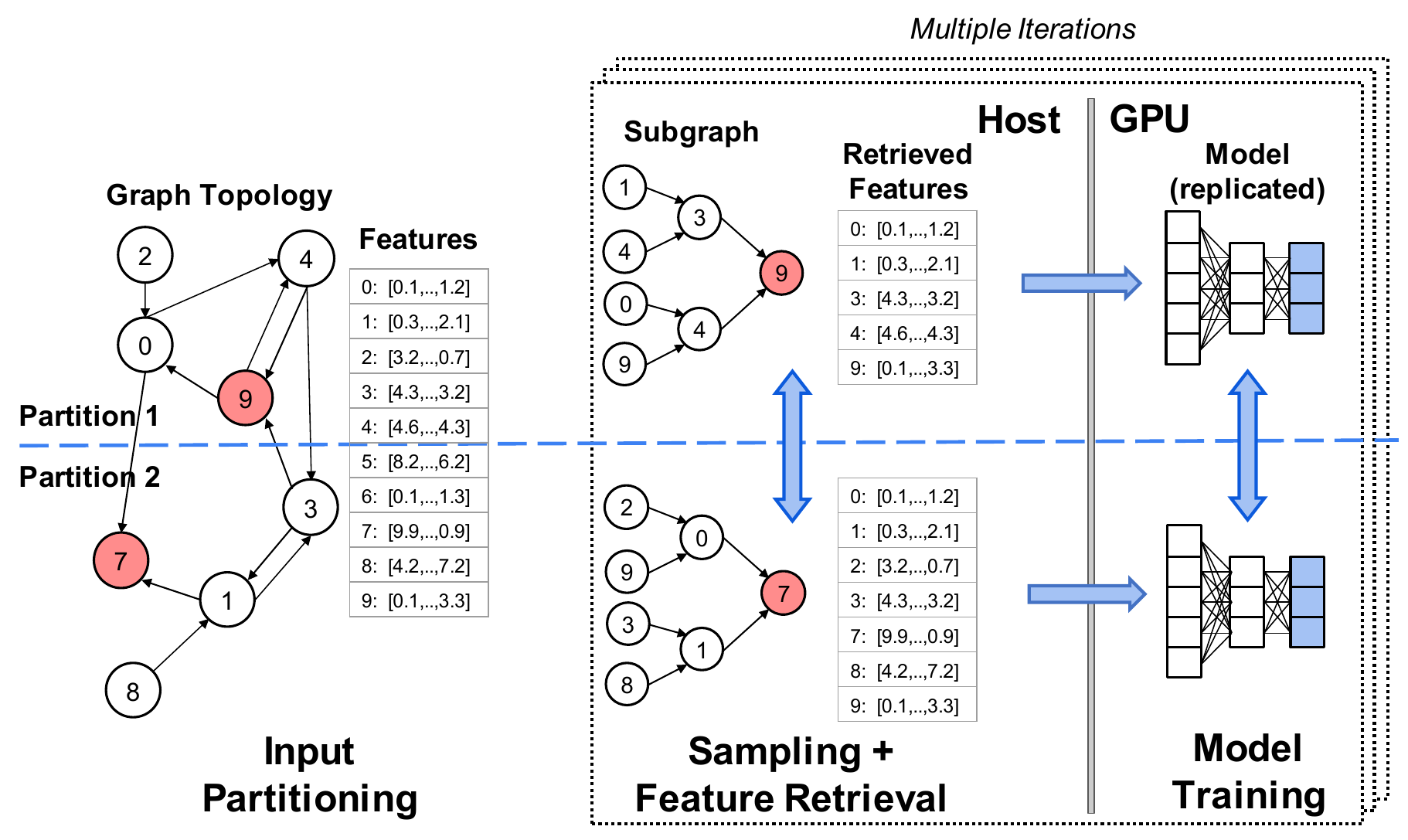}
 \caption{Sampling-based data-parallel GNN training. Each iteration samples multi-hop neighborhoods \emph{across partitions} and aggregates gradients. Following cross-partition edges triggers remote feature/activation fetches, which dominate per‑iteration cost at scale.}
 \label{fig:training_pipeline}
\end{figure}

\paragraph{Graph Partitioning}
Training a GNN on a large graph requires splitting the topology and features into partitions that are processed individually, as shown in Figure~\ref{fig:training_pipeline}. This data-parallel training approach replicates the model and aggregates gradients. However, when the graph is partitioned, message passing must follow edges that cross partitions, forcing remote feature/hidden-state fetches during each iteration. Cross-partition neighbor traffic (features/activations) grows with graph size, number of partitions, and depth, and quickly becomes a significant bottleneck~\cite{DBLP:conf/osdi/GandhiI21,DBLP:journals/pvldb/ZhangHL0HSGWZ020,bgl}. Table~\ref{tab:overhead} quantifies the magnitude for representative datasets; at 10~Gbps it implies 2--30 seconds of transfer per partition per epoch, and even at 400~Gbps the 37.7~GB case adds about 0.76~s per epoch, which is non-trivial relative to GPU compute. Moreover, computing partitions that both balance the load and reduce edge cuts is itself expensive and difficult~\cite{balanced-graph-partitioning}, and any imbalance further amplifies the effects of communication and stragglers.

\begin{table}[h]
 \centering
 \caption{Per-partition neighbor traffic during a single training epoch (3-layer model) versus number of partitions. Datasets are in Table~\ref{tab:datasets}. Inputs are $<\!1$\,GB.}
 \label{tab:overhead}
 \setlength{\tabcolsep}{0.5em}
 \small{
 \begin{tabular}{|l|l|r|r|r|}
 \hline
 \textbf{Partition} & \textbf{Dataset} & \textbf{2 parts.} & \textbf{4 parts.} & \textbf{8 parts.} \\
 \hline
 \multirow{2}{*}{\textbf{Random}} & Reddit & 6,106 MB & 9,216 MB & 10,832 MB \\
 \cline{2-5}
 & OGBN-Pr & 21,448 MB & 32,160 MB & 37,682 MB \\
 \hline
 \multirow{2}{*}{\textbf{METIS}} & Reddit & 3,264 MB & 5,180 MB & 5,512 MB \\
 \cline{2-5}
 & OGBN-Pr & 5,000 MB & 10,120 MB & 11,072 MB \\
 \hline
 \end{tabular}
 }
\end{table}

\paragraph{Why Fixes Fall Short}
Existing solutions help but do not eliminate remote fetches. High-quality partitions (e.g., METIS~\cite{metis}) reduce edge cuts but add preprocessing and memory overhead. Sampling lowers compute but not remote neighbor retrieval~\cite{fastgcn,case-for-sampling}. Fast interconnects (NVLink/RDMA)~\cite{DBLP:conf/usenix/SunSSSWWZLYZW23,dgcl,DBLP:conf/nsdi/LiuC00ZHPCCG23} help within a node, yet gains do not generalize: when traffic leaves NVLink islands, epoch times balloon ($>11\times$ on OGBN-Pa, $16\times$ on Reddit~\cite{mgg,DBLP:conf/nips/HamiltonYL17}). Uniform NVLink-class bandwidth across nodes demands specialized switching, which is costly. Caching~\cite{DBLP:conf/osdi/GandhiI21,DBLP:journals/pvldb/ZhangHL0HSGWZ020,ugache} reduces traffic but consumes HBM/DRAM. \ul{Thus, while existing approaches offer partial relief, the core tension between high capacity and cross-partition network traffic remains.}

Therefore, we move neighbor information exchange off the iteration critical path by training partitions in isolation and synchronizing only gradients; we then restore coverage statistically via repartitioning and correction.

\section{The \systemname{} Design}
\label{sec:design}

We begin with an overview of the system and its design principles (Section~\ref{subsec:design-overview}). Then, we introduce \systemname{}'s isolated training strategy (Section~\ref{subsec:design-sampling}), its approach to dynamic repartitioning (Section~\ref{subsec:design-repartitioning}), its sampling bias correction (Section~\ref{subsec:design-weighted-aggregation}), the training controller along with its heuristic for switching partitions (Section~\ref{subsec:design-controller}), and its phase-parallel training mode (Section~\ref{subsec:design-phase-training}).

\subsection{Overview}
\label{subsec:design-overview}

\systemname{} aims to efficiently train GNNs on large input graphs in a distributed, data-parallel fashion. \systemname{} supports both full-graph and sampling-based training modes (Section~\ref{sec:background}), although we focus mostly on sampling-based training as it offers superior scalability. \systemname{} divides the input graph and its features into partitions. Each partition is stored in CPU memory and loaded into GPU for training alongside the model. In the sampling-based training mode, we use the CPU to sample from the training nodes in the local partition to form mini-batches and load each batch into the GPU.

Unlike conventional GNN training, whether executing in full-graph or sampling mode, \systemname{} restricts graph access to the local partition only, preventing any data exchange along edges spanning multiple partitions, thereby executing training on each partition in complete isolation (Section~\ref{subsec:design-sampling}). This strategy \emph{eliminates cross-partition exchange of features or activations within an iteration}; cross-partition communication is \emph{gradient-only}, which yields significant performance gains.

\systemname{} supports all GNN models compatible with other state-of-the-art systems, maintaining similar or improving accuracy but with lower training time. These speedups are primarily due to reduced communication and synchronization overheads, as we will show in Section~\ref{sec:evaluation}. Moreover, due to its unique approach of processing partitions independently, \systemname{} achieves high flexibility in managing the overall training process. Since training on each partition can be performed separately from the others, there is no longer a need for concurrent data-parallel training of all partitions. Instead, it is enough for training on partitions to occur \emph{conceptually in parallel}. Within each phase, gradients are synchronized across the active replicas, and across phases, we carry optimizer state forward. As a result, depending on available resources, \systemname{} can train some partitions sequentially, seamlessly trading training time for capacity (Section~\ref{subsec:design-phase-training}).

\begin{figure}[htbp]
 \centering
 \begin{subfigure}[b]{0.47\textwidth}
 \includegraphics[width=1\textwidth]{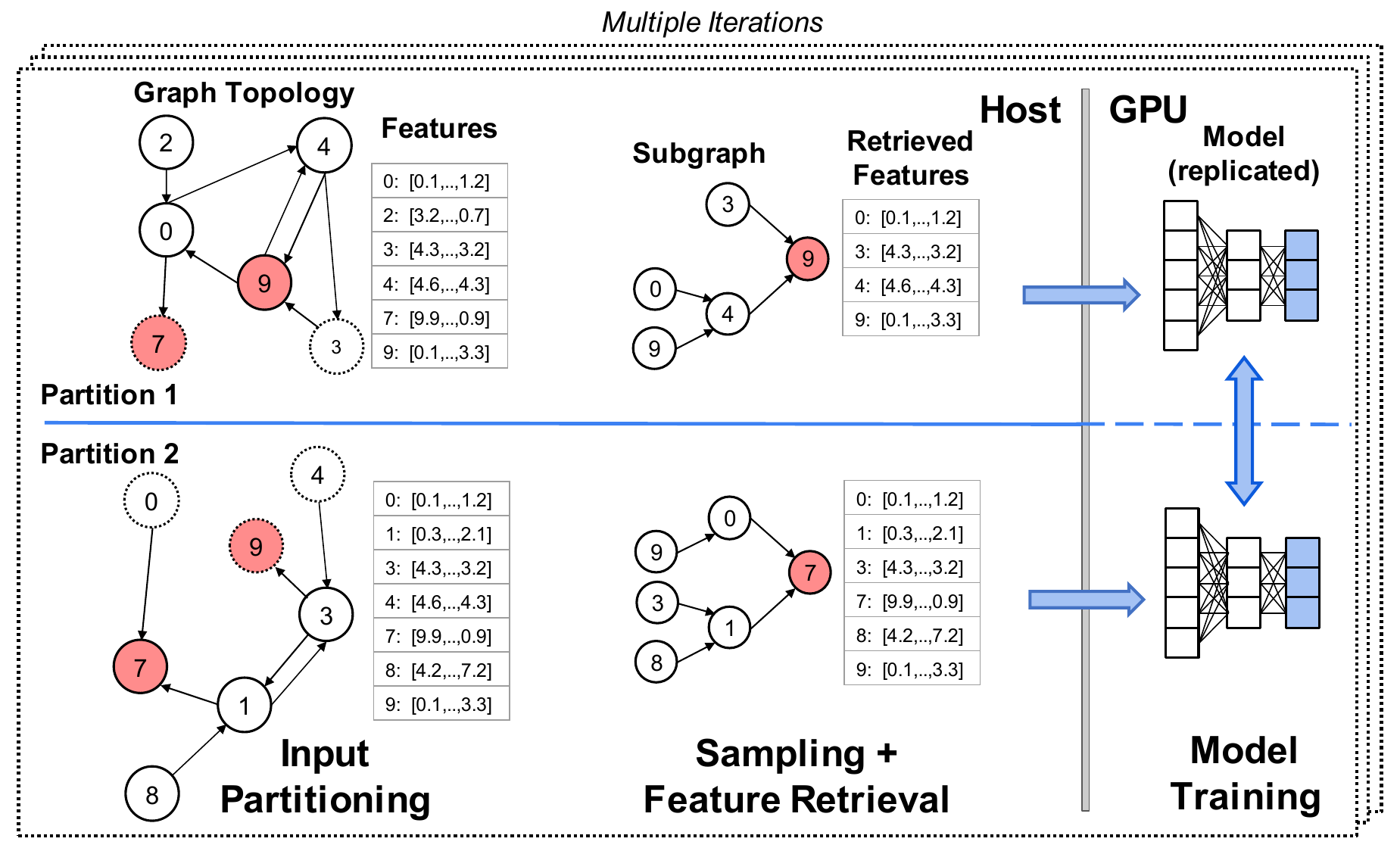}
 \caption{Super-epoch 1}
 \label{fig:sub1}
 \end{subfigure}
 \hfill
 \begin{subfigure}[b]{0.47\textwidth}
 \includegraphics[width=1\textwidth]{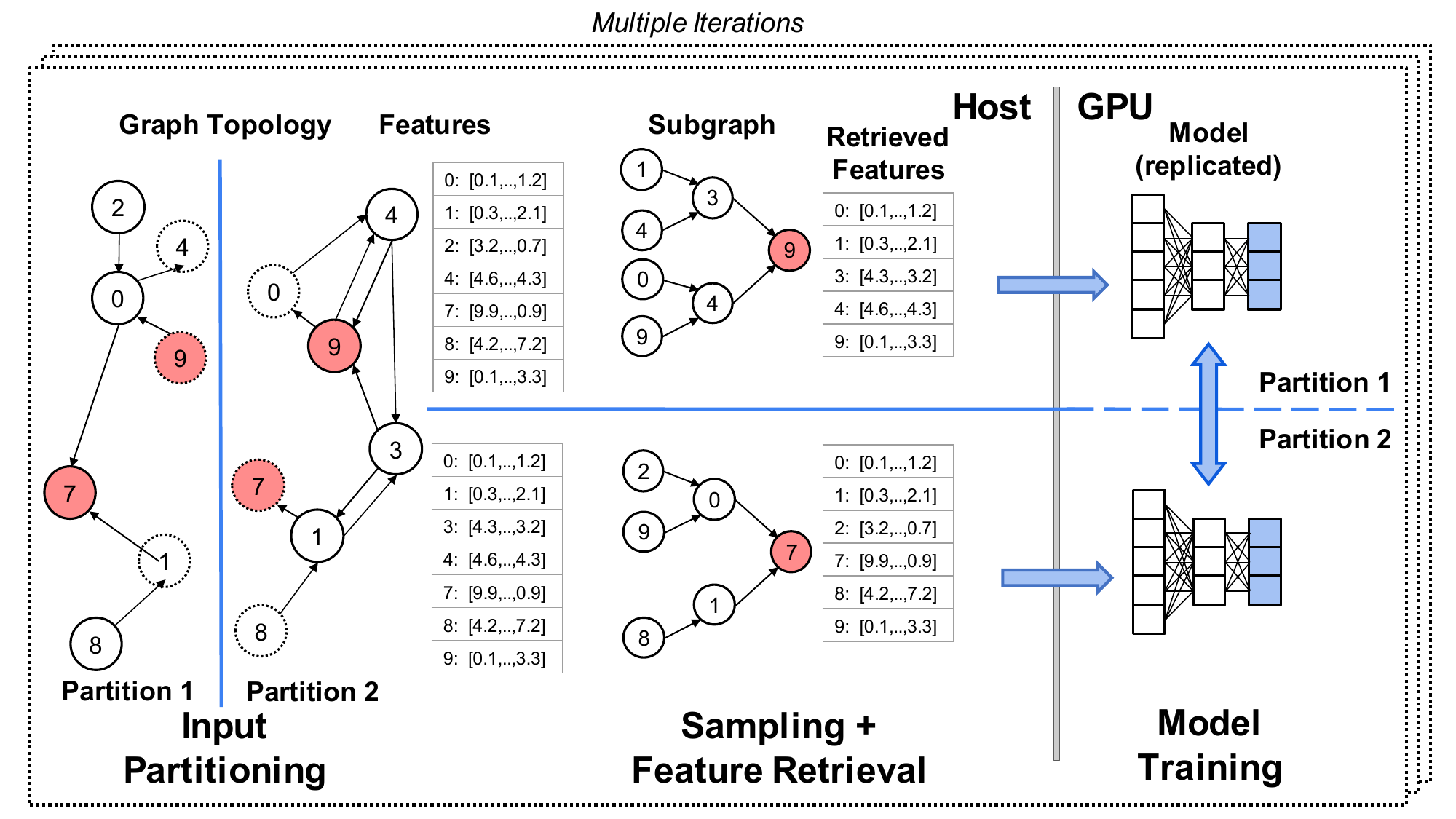}
 \caption{Super-epoch 2}
 \label{fig:sub2}
 \end{subfigure}
 \caption{Example of GNN training with \systemname{} for two super-epochs. Compared to Figure~\ref{fig:training_pipeline}, subgraph sampling is performed \emph{without remote neighbor fetches} (gradient-only communication) in each iteration and epoch. After a few epochs, we switch to a new super-epoch, creating a new partitioning layout (repartitioning).}
 \label{fig:training-example}
\end{figure}

Figure~\ref{fig:training-example} illustrates the workflow of \systemname{} as it trains a two-layer GNN model on an example input graph in sampling mode. Training is organized in super-epochs, each consisting of multiple training epochs. The input graph and its features are partitioned at the start of each super-epoch. We keep the partitioning layout unchanged for the entire super-epoch. Training proceeds in epochs that iterate over the whole set of training nodes in steps (also referred to as iterations). Within each step, we construct mini-batches of subgraphs and associated features. The GPU then processes each mini-batch as input for model training. Finally, the gradients are accumulated across all replicated models for each partition before proceeding to the next iteration.

\systemname{} allows each node to be exposed to its entire neighborhood despite its partition-level isolation by strategically repartitioning the input between super-epochs (Section~\ref{subsec:design-repartitioning}). Repartitioning the input several times throughout the training to disseminate neighborhood information is more cost-effective than exchanging data across partitions. The rationale behind this design decision is that, with conventional training, the total neighbor‑exchange volume scales roughly with the sum of the fanout at each layer and grows rapidly with model depth and cut edges. In contrast, repartitioning cost is linear in the input size and the number of super‑epochs.

Repartitioning alone is insufficient to ensure high model accuracy as, despite exposing each vertex to its entire neighborhood throughout all super-epochs, we may still introduce bias into the sampling process unwillingly. For example, during two super-epochs, a neighbor could be seen more or less frequently due to the partition structure and the presence or absence of an edge. \systemname{} accounts for sample bias by keeping track of how many times a vertex is seen during a super-epoch and scaling the gradients for that partition before aggregation. We introduce a low-overhead weighted aggregation scheme based on estimations to adjust for sampling bias (Section~\ref{subsec:design-weighted-aggregation}). We directly use this scheme in the \systemname{} training controller to decide when to repartition the graph, switching to the next super-epoch (Section~\ref{subsec:design-controller}).

Intuitively, isolated training may act as a form of dropout in the input layer, which can improve generalization for deeper models. Repartitioning restores long‑run neighbor coverage, and the coverage‑corrected scaling re-centers gradients toward the full‑graph objective.

Together, these mechanisms preserve the convergence properties akin to conventional GNN training, enabling the model to effectively learn from the global graph structure while reducing communication overhead.

\subsection{Isolated Training Strategy}
\label{subsec:design-sampling}

Unlike conventional GNN training systems that follow edges to retrieve nodes and features from other partitions, \systemname{} operates entirely within each local partition and uses this isolation to streamline training.

\paragraph{Isolated Training} \systemname{} proceeds similarly to traditional training, but considers exclusively local nodes and edges. \emph{Halo nodes} are the key enabler: they cache boundary neighbors so that a partition contains all nodes within $k$ hops of any training node, allowing $k$-layer message passing without remote fetches. Halos are read-only within a super-epoch and contribute to local aggregations but never trigger cross-partition communication. With halos in place, \systemname{} can proceed as if each partition executes message passing on its local subgraph, without remote dependencies within an iteration. In each iteration, we either feed the entire partition to the model (full-graph mode) or select a subset of the training nodes from the local partition to construct subgraphs (sampling mode). Each subgraph's depth aligns with the number of layers in the model; for a $k$-layer model we sample $k$-hop subgraphs, and halo nodes provide $k$-hop closure within each partition so each mini-batch is self-contained without cross-partition reads.

\paragraph{Overhead Reduction} The primary advantage of \systemname{}'s approach is reducing communication overhead and synchronization delays. In traditional GNN training, resource consumption during sampling grows rapidly with graph size, model depth, and partition count. By eliminating cross-partition exchange of features or activations during a training iteration, \systemname{} reduces the time spent on \emph{cross-partition neighbor exchange} to zero, leaving only gradient synchronization. The system still accumulates and updates gradients across partitions after each iteration to ensure the model parameters benefit from each partition. This approach enables both an accelerated training process and enhanced scalability. Isolation introduces a predictable \emph{coverage bias}: a node’s gradient in a step reflects only neighbors present in the current partition. We correct this bias in aggregation using a single scalar per batch, derived from easy-to-measure degree ratios (Section~\ref{subsec:design-weighted-aggregation}).

\subsection{Dynamic Graph Repartitioning}
\label{subsec:design-repartitioning}

\systemname{} dynamically manages graph partitions throughout training by repartitioning the graph to expose nodes to new neighborhoods over time. This differs from other GNN training systems that use static partitioning without isolation. Adapting the partitioning layout improves model quality.

\textbf{Terminology.} We use \emph{iteration} to mean one mini-batch step; an \emph{epoch} is a full pass over training nodes under a fixed partition layout; a \emph{super-epoch} is a contiguous group of epochs between repartitioning events. A \emph{chunk} is a static unit of graph data produced once at startup. A \emph{partition} is the training unit formed by combining a base chunk with a swept chunk for a super-epoch. A \emph{worker} is an execution resource (GPU) responsible for one base chunk that trains one partition at a time. A \emph{phase} is a scheduling unit when only a subset of partitions run concurrently (Section~\ref{subsec:design-phase-training}).

\paragraph{Super-epochs}

\systemname{} structures the overall training process into super-epochs. A super-epoch is simply a contiguous span of epochs between repartitionings and is independent of the number of workers. The training controller (Section~\ref{subsec:design-controller}) makes the decision to transition to a new super-epoch where the input graph will undergo repartitioning. The new partitions must be distinct from the partitions in the previous super-epoch to ensure that each vertex is exposed to its entire neighborhood. Once the repartition is complete, we begin executing the next super-epoch. The challenge with repartitioning is that recomputing partitions and shuffling data between each super-epoch may be prohibitively expensive and negate the performance benefits of isolated training.

\paragraph{Graph Partitioning}

Many graph partitioning strategies exist, such as random or METIS~\cite{metis}. In random graph partitioning, nodes in the input graph are randomly assigned to different partitions without particular consideration for the underlying graph structure. In contrast, METIS graph partitioning recursively coarsens the graph to reduce its size, partitions the smaller graph, and then uncoarsens it to obtain a balanced and high-quality partition of the original graph. Upon partitioning the nodes, the edges linked to these nodes and their corresponding endpoints are also included in the respective partition. Should an endpoint not be part of the initial partition, it is incorporated as a halo node, represented with dashed lines in Figure~\ref{fig:training-example}. Finally, the features associated with each node are added to the partition.

\paragraph{Chunks and Efficient Partition Construction}

Repartitioning the graph introduces overhead: nodes and edges must be assigned to partitions, and feature data (typically vectors of length 128 to 1024) must be split. To reduce this overhead, \systemname{} uses unmodified graph partitioning strategies such as random or METIS to split the dataset into \emph{chunks}, rather than final partitions. This initial partitioning is performed only once. Each partition for training is dynamically constructed by combining pairs of chunks, specifically a base chunk, $m_i$, with another variable chunk, $m_j$ where $i \neq j$. This approach allows for quick and efficient repartitioning of the graph by sweeping through different variable chunks without recomputing a new set of partitions after each super-epoch. Each worker is responsible one chunk and sweeps one chunk from the other chunks during each repartitioning, thus only loading a small subset of the graph each time. This method effectively disseminates neighborhood information across the graph structure. Figure~\ref{fig:sweep_chunk} shows a simple example of \systemname{}'s efficient sweeping chunk partitioning. This coverage holds for $k$-layer models as the locality boundary changes across super-epochs, ensuring that neighbors within graph distance $\le k$ appear with non-zero probability over time.

\begin{figure}
 \raggedleft
 \includegraphics[width=0.47\textwidth]{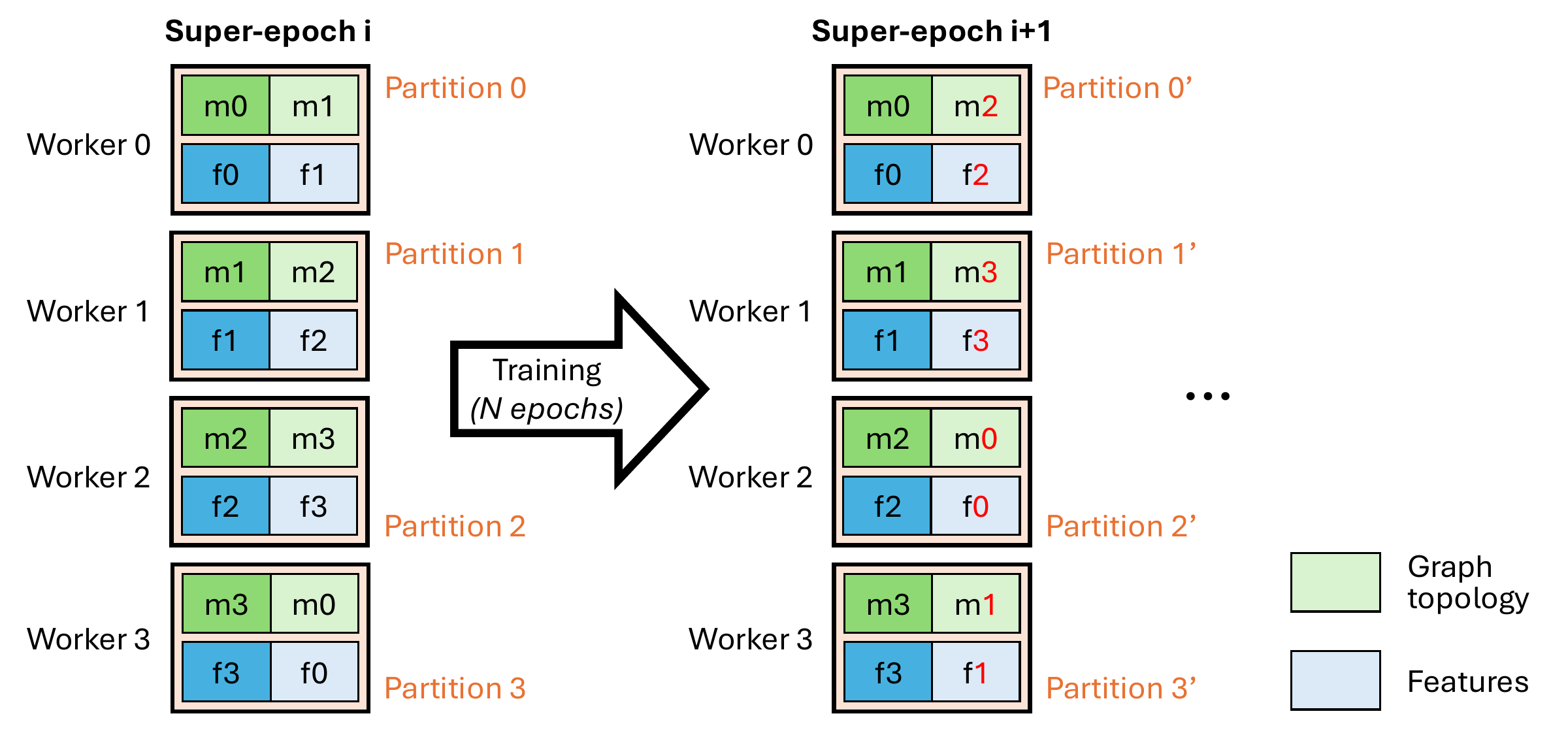}
 \caption{Example of efficient dynamic repartitioning via sweeping chunks. The example shows four workers operating in a data-parallel manner. The graph topology and feature data are divided into chunks {$m_0$, $m_1$, $m_2$, $m_3$} and {$f_0$, $f_1$, $f_2$, $f_3$} respectively. The dataset is split into chunks just once before training. The four workers then operate in parallel. Worker 0 initially loads chunks $m_0$ and $f_0$, then loads $m_1$ and $f_1$, groups these together as a partition, and trains the model on the resulting partition for a super-epoch. After that, we perform a repartition where worker 0 will load the next chunks $m_2$ and $f_2$, grouping them with $m_0$, $f_0$ to form a new partition for further training. This process repeats until all workers have seen all chunks at least once.}
 \label{fig:sweep_chunk}
\end{figure}

\paragraph{Coverage guarantee.} With $M+1$ chunks and $W$ workers, each worker pins a base chunk and cycles through others across super-epochs. Over $M$ super-epochs every cross-chunk edge co-resides in at least one partition, satisfying Theorem~\ref{thm:unbiased}'s support requirement. When $W < M$, base chunks rotate in round-robin so all pairs are covered, making inclusion frequencies well-defined for bias correction.

\subsection{Coverage-Corrected Aggregation}
\label{subsec:design-weighted-aggregation}

\paragraph{Problem: Local Sampling Bias} When we train in isolation, a node $v$ only sees its \emph{local} neighbors in the current partition. The gradient we compute in that step is, therefore, biased toward the local neighborhood. This can be seen as a form of biased sampling and our goal is to correct for it so that model updates reflect the full graph. Our procedure is inspired by FastGCN~\cite{fastgcn} but differs in its approach. FastGCN corrects sampling at a single layer, while we correct the full gradient after all layers, so the correction is independent of depth. In addition to directly addressing the quantity of interest in training, our approach has the advantage of being model-agnostic, as it does not require modifying activations or other model components but only operates on gradients.

\paragraph{Solution: Bias-Corrected Estimators} We propose a set of estimators with different properties that allow tailoring the solution to different use cases. We first propose a node-level estimator that eliminates sampling bias exactly. We then propose a coarser estimator that operates on batch-level gradient objects, enabling compatibility with common software packages such as PyTorch~\cite{pytorch}. Finally, we propose a shrinkage-based estimator that consciously accepts some bias to reduce overfitting.

\paragraph{Notation} Consider a model represented by a function $g_v(u,\theta)$, parametrized by a parameter vector $\theta\in\mathbb{R}^{|\theta|}$, that gives the loss contribution for a node $v$ of a single neighboring node $u$. We omit $\theta$ for readability. Let $\mathcal{D}$ be the distribution of nodes in the original graph and $\mathcal{D}_v$ the conditional distribution of neighbors of $v$ in the original graph.

\paragraph{Full-Graph Unbiased Estimator} The true gradient of the loss function $L$ is given by:
\begin{equation}\label{eq:gradient}
\nabla_\theta L = \mathbb{E}_{v\sim \mathcal{D}}
\left[
 \mathbb{E}_{u\sim \mathcal{D}_v}
 \left[
 \nabla_\theta g_v
 \left(
 u
 \right)
 \right]
\right]
\end{equation}
where the expectation over $\mathcal{D}_v$ models neighborhood aggregation. This can be estimated by the unbiased estimator:
\begin{equation}
 \nabla_\theta \tilde{L} = \frac{1}{|S_\mathcal{D}|} \sum_{v\in S_\mathcal{D}} \frac{1}{|S_{\mathcal{D}_v}|} \sum_{u\in S_{\mathcal{D}_v}} \nabla_\theta g_v(u)
\end{equation}
where the samples $S_\mathcal{D}$ and $S_{\mathcal{D}_v}$ are i.i.d.\ draws according to the original distributions $\mathcal{D}$ and $\mathcal{D}_v$.

\paragraph{Sampling Bias in Isolated Training} Viewed across partitions, nodes are still sampled from the full distribution $\mathcal{D}$ in the isolated training strategy. Neighboring nodes, however, are drawn from the conditional distribution $\mathcal{D}_v^\text{local}$ of neighboring nodes of $v$ present in the local partition.

\paragraph{Importance Weighting} We want to re-weight gradient contributions to recover the true gradient $\nabla_\theta L$ \eqref{eq:gradient}. Our guiding idea will be a well-known importance sampling result~\cite{glynn1989importance} that allows recovering the original distribution. Denoting the probabilities of selecting a neighboring node $u$ under the original distribution $\mathcal{D}_v$ and under the local distribution $\mathcal{D}_v^\text{local}$ by $p_v(u)$ and $q_v(u)$, respectively, we have:
\begin{equation} \label{eq:importance_sampling}
\begin{split}
& \mathbb{E}_{v\sim\mathcal{D}} \left[ \mathbb{E}_{u\sim\mathcal{D}^\text{local}_v} \left[ \frac{p_v(u)}{q_v(u)} \nabla_\theta g_v(u) \right] \right] = \\
& \mathbb{E}_{v\sim\mathcal{D}} \left[ \mathbb{E}_{u\sim\mathcal{D}_v} \left[ \nabla_\theta g_v(u) \right] \right] =
 \nabla_\theta L
\end{split}
\end{equation}
if $q_v(u) > 0$ whenever $p_v(u) > 0$.

\paragraph{Asymptotically Unbiased Node-Level Estimator} We cannot apply the importance sampling result \eqref{eq:importance_sampling} directly because the probability of selecting a neighbor missing from the local partition is zero, violating the support condition. We resolve this via repartitioning: as long as the long-run probability is positive, the bias can be removed entirely, at least asymptotically. The sweep schedule presented in Section~\ref{subsec:design-repartitioning} guarantees that each node will be matched with all of its neighbors already after a finite number of super-epochs, so the assumption is justified. Formally, let $q_v^t(u)=\mathbb{P}\left[u\in S_{\mathcal{D}^\text{local}_v} \right]$ be the probability of selecting neighbor $u$ of $v$ in the current local partition at super-epoch $t$. We define the estimator:
\begin{equation}
\begin{split}
& \nabla_\theta \tilde{L}^{\text{corr}} = \\
& \frac{1}{T}\sum_{t=1}^T\frac{1}{|S_\mathcal{D}^t|} \sum_{v\in S_\mathcal{D}^t} \frac{1}{|S_{\mathcal{D}^\text{local}_v}|} \sum_{u\in S_{\mathcal{D}^\text{local}_v}} \frac{p_v(u)}{q_v^t(u)} \nabla_\theta g_v(u)
\end{split}
\end{equation}
where $T$ is the number of super-epochs and $|S_\mathcal{D}^t|$ is the local node sample in the current super-epoch. If the partitioning gives enough weight to each node such that:
\begin{equation}\label{eq:convergence_q}
\lim_{T\rightarrow\infty} \frac{1}{T} \sum_{t=1}^T \mathbb{I}_{u\in S_{\mathcal{D}_v^\text{local}}} \xrightarrow{\text{a.s.}}q_v(u) >0
\end{equation}
i.e., the long-run selection probability is positive, we have:
\begin{theorem}\label{thm:unbiased}
Under assumption~\eqref{eq:convergence_q} and if $\left| \frac{p_v(u)}{q_v^t(u)} \nabla_\theta g_v(u)\right|\le M$ for some constant $M$ and $S_{\mathcal{D}}^t$ and $S_{\mathcal{D}^\text{local}_v}$ are i.i.d.\ samples,
\begin{equation}
\lim_{T\rightarrow\infty} \mathbb{E}\left[ \nabla_\theta \tilde{L}^{\text{corr}} \right] = \nabla_\theta L
\end{equation}
\end{theorem}
\begin{proof}
By i.i.d.\ sampling, we have:
\begin{equation}
\mathbb{E}\left[ \nabla_\theta \tilde{L}^{\text{corr}} \right] = \frac{1}{T} \sum_{t=1}^T \mathbb{E}_{v\sim\mathcal{D}} \left[ \mathbb{E}_{u\sim\mathcal{D}^\text{local}_v} \left[ \frac{p_v(u)}{q^t_v(u)} \nabla_\theta g_v(u) \right] \right]
\end{equation}

By the dominated convergence theorem and assumption~\eqref{eq:convergence_q}, we have:
\begin{equation}
\begin{split}
& \lim_{T\rightarrow\infty} \frac{1}{T} \sum_{t=1}^T \mathbb{E}_{v\sim\mathcal{D}} \left[ \mathbb{E}_{u\sim\mathcal{D}^\text{local}_v} \left[ \frac{p_v(u)}{q^t_v(u)} \nabla_\theta g_v(u) \right] \right] = \\
& \mathbb{E}_{v\sim\mathcal{D}} \left[ \mathbb{E}_{u\sim\mathcal{D}^\text{local}_v} \left[ \frac{p_v(u)}{q_v(u)} \nabla_\theta g_v(u) \right] \right]
\end{split}
\end{equation}

The result follows by applying the importance sampling theorem~\eqref{eq:importance_sampling}.
\end{proof}
That is, $\nabla_\theta \tilde{L}^{\text{corr}}$ is asymptotically unbiased under mild regularity conditions. An important special case is when all neighbors have equal weight under $\mathcal{D}_v$ for all $v$. Then $p_v(u) = 1/d_v^\text{global}$, $q^t_v(u) = 1 / d_v^\text{local}$ for all local neighbors $u$ of $v$, where $d_v^\text{global}$ and $d_v^\text{local}$ are the degrees of node $v$ in the full graph and in the local partition, respectively. So $\nabla_\theta \tilde{L}^{\text{corr}}$ becomes:
\begin{equation}
\begin{split}
& \nabla_\theta \tilde{L}^{\text{uniform}} = \\
& \frac{1}{T}\sum_{t=1}^T\frac{1}{|S_\mathcal{D}^t|} \sum_{v\in S_\mathcal{D}^t} \frac{1}{|S_{\mathcal{D}^\text{local}_v}|} \frac{d_v^{\text{local}}}{d_v^\text{global}} \sum_{u\in S_{\mathcal{D}^\text{local}_v}} \nabla_\theta g_v(u)
\end{split}
\end{equation}
In general one may use $q^t_v(u) = p_v(u)/\sum_{i\in \mathcal{D}^{\text{local}}_v} p_v(i)$.
\paragraph{Batch-Level Approximation} Exact bias correction requires modifying gradient contributions at the level of individual nodes in a batch. Common libraries such as PyTorch \cite{pytorch} only provide hooks for modifying batch-level gradient objects, however. In line with the idea of being model-agnostic, it would be desirable to have a framework that can be used with PyTorch~\cite{pytorch} and its batch-level gradient objects. We propose the following batch-level estimator and will proceed to show that it is closest (in L2-distance of expectations) to the asymptotically unbiased estimator under mild assumptions:
\begin{equation}\label{eq:batch-estimator}
\begin{split}
& \nabla_\theta \tilde{L}^{\text{batch}}(c) = \\
& \frac{1}{T}\sum_{t=1}^T\frac{c^t}{|S_\mathcal{D}^t|} \sum_{v\in S_\mathcal{D}^t} \frac{1}{|S_{\mathcal{D}^\text{local}_v}|} \sum_{u\in S_{\mathcal{D}^\text{local}_v}} \nabla_\theta g_v(u)
\end{split}
\end{equation}
where the batch-level correction factors $c\in\mathbb{R}^T$ are:
\begin{equation}\label{eq:batch_correction_factor}
c^t = \frac{1}{|S_\mathcal{D}^t|} \sum_{v\in S_\mathcal{D}^t} \frac{1}{|S_{\mathcal{D}^\text{local}_v}|} \sum_{u\in S_{\mathcal{D}^\text{local}_v}} \frac{p_v(u)}{q_v^t(u)}
\end{equation}
In the case of equal selection probabilities, this simplifies to:
\begin{equation}\label{eq:correction_uniform}
c^t_\text{uniform} = \frac{1}{|S^t_{\mathcal{D}}|} \sum_{v\in S^t_{\mathcal{D}}} \frac{d_v^{\text{local}}}{d_v^\text{global}}
\end{equation}
In practice, we apply the batch-level correction $\nabla_\theta \tilde{L}^{\text{batch}}(c)$ \emph{before} the phase's gradient synchronization. We have (for notational convenience, we set $T=1$ without loss of generality and drop the $t$-superscript):
\begin{theorem}\label{thm:minbiased}
Assuming that gradients are independently identically distributed $\nabla_\theta g_v(u)\sim \text{i.i.d}$ with finite mean $\mu\in \mathbb{R}^{|\theta|}\setminus\mathbf{0}$ independent of $\frac{p_v(u)}{q_v(u)}$, equation~\eqref{eq:batch_correction_factor} gives the scalar $c$ that minimizes:
\begin{equation}\label{eq:biasobjective-appendix}
\left\Vert \mathbb{E} \left[ \nabla_\theta \tilde{L}^{\text{corr}} \right] - c \mathbb{E} \left[ \nabla_\theta g_{S_\mathcal{D}} \right] \right\Vert_2
\end{equation}
\end{theorem}
\begin{proof}
After squaring the objective function~\eqref{eq:biasobjective-appendix}, standard vector calculus gives the orthogonal projection of $\mathbb{E} \left[ \nabla_\theta \tilde{L}^{\text{corr}} \right]$ onto $\mathbb{E} \left[ \nabla_\theta g_{S_\mathcal{D}} \right]$:

\begin{equation}
\begin{split}
& \frac{\partial}{\partial c} \left\Vert \mathbb{E} \left[ \nabla_\theta \tilde{L}^{\text{corr}} \right] - c \mathbb{E} \left[ \nabla_\theta g_{S_\mathcal{D}} \right] \right\Vert_2^2 = 0 \Rightarrow c = \\
& \left( \mathbb{E} \left[ \nabla_\theta g_{S_\mathcal{D}}\right]^\intercal \mathbb{E} \left[\nabla_\theta g_{S_\mathcal{D}} \right]\right)^{-1} \mathbb{E} \left[ \nabla_\theta g_{S_\mathcal{D}}\right]^\intercal \mathbb{E} \left[\nabla_\theta \tilde{L}^{\text{corr}} \right]
\end{split}
\end{equation}

By the assumptions of the theorem:

\begin{align}
\mathbb{E} \left[ \nabla_\theta g_{S_\mathcal{D}} \right] & = \frac{1}{|S_\mathcal{D}|} \sum_{v\in S_\mathcal{D}} \frac{1}{|S_{\mathcal{D}^\text{local}_v}|} \sum_{u\in S_{\mathcal{D}^\text{local}_v}} \mu = \mu \\
\begin{split}
\mathbb{E} \left[ \nabla_\theta \tilde{L}^{\text{corr}} \right] & = \frac{1}{|S_\mathcal{D}|} \sum_{v\in S_\mathcal{D}} \frac{1}{|S_{\mathcal{D}^\text{local}_v}|} \sum_{u\in S_{\mathcal{D}^\text{local}_v}} \frac{p_v(u)}{q_v(u)} \mu \\ & = s_L \mu
\end{split}
\end{align}

where $s_L =\frac{1}{|S_\mathcal{D}|} \sum_{v\in S_\mathcal{D}} \frac{1}{|S_{\mathcal{D}^\text{local}_v}|} \sum_{u\in S_{\mathcal{D}^\text{local}_v}} \frac{p_v(u)}{q_v(u)}$. This gives:

\begin{equation}
\begin{split}
& c = \left( \mu^\intercal \mu\right)^{-1} \mu ^\intercal \mu s_L = s_L = \\
& \frac{1}{|S_\mathcal{D}|} \sum_{v\in S_\mathcal{D}} \frac{1}{|S_{\mathcal{D}^\text{local}_v}|} \sum_{u\in S_{\mathcal{D}^\text{local}_v}} \frac{p_v(u)}{q_v(u)}
\end{split}
\end{equation}
\end{proof}
That is, $\nabla_\theta \tilde{L}^{\text{batch}}$ is the estimator with batch-level correction whose expectation is closest (in $L^2$-distance) to the expectation of $\nabla_\theta \tilde{L}^{\text{corr}}$ under mild distributional assumptions.

\paragraph{Shrinkage} The estimators above assume idealized conditions. In practice, accepting some degree of bias to shrink coefficients can improve stability and combat overfitting \cite{james1961estimation}. We find that the following correction factor gives better accuracy than the minimally-biased version:
\begin{equation}\label{eq:resampling}
c^t_\text{resampling} = \frac{1}{\sum_{v\in S^t_{\mathcal{D}}}\left[ \frac{d_v^\text{global}}{d_v^{\text{local}}} -1 \right] |S_{\mathcal{D}^\text{local}_v}|}
\end{equation}
While $c^t$ uses an arithmetic mean of correction factors, $c^t_\text{resampling}$ is closer to a harmonic mean and produces lower correction factors overall. We interpret this as a form of shrinkage in gradient magnitude, similar to methods like gradient clipping and others that have been developed for deep learning applications \cite{pascanu2013difficulty,schulman2015trust}. Intuitively, $c^t_\text{resampling}$ weights the correction factors inversely by how often, in expectation, the global set of neighbors of a selected node would have to be resampled to get a sample consisting only of neighbors present in the local partition.

\subsection{Training Controller and Partitions Switching}
\label{subsec:design-controller}

Having established how to correct gradients within a super-epoch, we now describe how \systemname{} decides when to switch partitions.

The training process in \systemname{} is managed and monitored by a training controller that collects information on the overall training progress on each machine and for each partition. The controller aggregates sampling information within each partition and initiates repartitioning when appropriate.

The controller currently aims to ensure that each possible combination of chunks for a partition are explored for a fixed number of iterations, providing it with a target of $e / (c-1)$ super-epochs where $e$ is a fixed number of epochs to train for and $c$ is the number of chunks in the system. The controller monitors a simple \emph{coverage deficit} statistic per partition, $\Delta_t = 1 - \widehat{c}_t$, and triggers an earlier repartition when $\Delta_t$ persists across several steps, prioritizing partitions with low coverage. Gradient aggregation is performed synchronously within the current phase; repartitioning only updates metadata and does not introduce neighbor traffic. We empirically verify that this simple heuristic is sufficient for our needs (Section~\ref{sec:evaluation}). Therefore, we leave the implementation of more sophisticated switching heuristics for future work.

\subsection{Phase-Parallel Training}
\label{subsec:design-phase-training}

Gradient-only communication unlocks an additional benefit: flexible scheduling. Since partitions exchange no features or activations within an iteration, they need not run concurrently. \systemname{} exploits this independence to allow the sequential execution of partitions in phases. This strategy significantly amplifies training capacity, which benefits large input graphs that do not fit into the collective memory of all available machines.

When accommodating concurrent training on all partitions is unfeasible, \systemname{} can adopt a phase-parallel training approach that divides the partitions into distinct phases, each involving the execution of only a subset of the total partitions. Each phase runs standard data-parallel SGD over its $M$ active partitions: gradients are all-reduced within the phase, parameters are updated once, and the updated parameters and optimizer state are carried into the next phase. Over one epoch, all phases are executed so every partition contributes exactly one update per epoch, just serialized in time.

Algorithm~\ref{alg:phase-parallel} details phase-parallel training. By default $P = M$ (all partitions concurrent), but $M < P$ serializes phases; $M = 1$ trains arbitrarily large graphs on one machine.

\begin{algorithm}[h]
\caption{Phase Parallel Training in \systemname{}}
\label{alg:phase-parallel}
\KwIn{Graph $G$, number of partitions $P$, max partitions per phase $M$}
\KwOut{Trained model parameters}

Initialize model parameters $\theta$, optimizer state\;
Partition $G$ into $P$ partitions\;

\For{each epoch}{
 \For{phase $i \gets 1$ \KwTo $\lceil P/M \rceil$}{
 active $\gets$ partitions for phase $i$\;
 \tcp{Each active partition processes all its batches}
 \For{each iteration (mini-batch) across active partitions}{
 \ForPar{each worker $w$ with partition in active}{
 sampled $\gets$ isolated\_sampling(partition)\;
 grad $\gets$ backward\_pass(forward\_pass(sampled))\;
 scale grad by coverage factor $c$\;
 }
 all\_reduce(grad) across active workers\;
 $\theta \gets$ optimizer\_step($\theta$, grad)\;
 }
 \tcp{Carry $\theta$ and optimizer state to next phase}
 }
}
\KwRet{$\theta$}
\end{algorithm}

\systemname{}'s phase-parallel training enables trading training time for memory capacity.

\balance

\section{Implementation}
\label{sec:implementation}

\systemname{} follows a client–server architecture that enforces gradient-only communication during iterations. Servers store the graph structure and node features while clients issue sampling requests and drive training. During iterations, servers never exchange features or activations across partitions; cross-server traffic consists only of gradient all-reduce. At super-epoch boundaries, we perform feature movement for halo synchronization and partition switching.

We run one server and one client per partition, co-located on the same machine for locality. Each client samples subgraphs from its local partition on the CPU and feeds batches to the GPU for forward and backward passes. Gradients are computed locally and then aggregated across the concurrently active partitions to update model parameters. This preserves the per-iteration invariant of gradient-only communication while supporting both fully parallel and phase-parallel execution.

\paragraph{Estimator choice.}
We use batch-level coverage-corrected aggregation $\nabla_\theta \tilde{L}^{\text{batch}}(c)$ (Equation~\eqref{eq:batch-estimator}) with the resampling-based factor $c^t_{\text{resampling}}$ (Equation~\eqref{eq:resampling}). This design integrates cleanly with PyTorch as a single per-batch gradient scaling hook (applied immediately before the all-reduce), adds negligible overhead, and avoids re-implementing node-level gradient accounting. The node-level variant would require substantial custom machinery in PyTorch, complicating fair performance comparisons against highly optimized baselines. The results for the minimally biased variant are shown in Section~\ref{sec:minbiased}.

\paragraph{Memory management.}
Partition data is memory-mapped from disk (addressable in CPU memory but paged in on demand), reducing peak memory. Halo node features are pre-cached at super-epoch boundaries via all-gather, eliminating cross-partition reads during iterations. GPU gradient buffers use pinned allocation for efficient transfers.

\paragraph{Repartitioning.}
Workers advance through neighbor chunks in round-robin order, with a configurable epoch interval between switches. At each switch, workers load the new chunk's edges and synchronize halo features before resuming; a barrier ensures all workers complete before the next super-epoch begins.

\paragraph{Communication.}
Gradient synchronization uses PyTorch's NCCL backend for all-reduce. Feature gathering during repartitioning uses all-to-all collectives. We build \systemname{} on top of DGL~\cite{dgl} and PyTorch~\cite{pytorch}; the implementation adds approximately 4,500 lines of Python and C++ on top of DGL's partitioned graph infrastructure.

\section{Evaluation}
\label{sec:evaluation}

Our primary emphasis is on model-level metrics (train accuracy, test accuracy) and system-level metrics (throughput, runtime, capacity). We focus on providing in-depth answers to the following research questions.
\begin{enumerate}[label=\textbf{RQ\arabic*:}]
 \item How does the performance of \systemname{} compare to state-of-the-art distributed training systems in terms of training time and accuracy? (Section~\ref{subsec:eval-overall-comparison})
 \item How does \systemname{} scale as the number of graph partitions or the graph size increases? In particular, can phase-parallel training enable the training of very large graphs that do not fit in a single machine? (Section~\ref{subsec:eval-scalability})
 \item What are the overheads and benefits of the various components of \systemname{}? (Section~\ref{subsec:eval-overheads-ablation})
\end{enumerate}

\subsection{Experimental Setup}
\label{subsec:eval-setup}

\paragraph{GNN Models}
We use three standard GNN models for our experiments, illustrating a varied set of graph-based tasks.

\begin{description}
\item[GraphSage (Sage)~\cite{DBLP:conf/nips/HamiltonYL17}] is a general inductive framework that efficiently generates low-dimensional node embeddings for previously unseen data. At each layer, the hidden states of a node's neighbors are first aggregated (for example, by averaging). These states and the states from the previous layer are combined and multiplied by a matrix. This process yields the hidden state at the current layer.
\item[Graph Convolutional Network (GCN)~\cite{DBLP:conf/iclr/KipfW17}] \hfill\break
is a semi-supervised node-classification method. Each layer applies a linear map to hidden states (or initial features) and averages neighbors with degree-based weights.
\item[Graph Attention Network (GAT)~\cite{DBLP:conf/iclr/VelickovicCCRLB18}] assigns weights to different neighborhood nodes using attention. This approach makes aggregating neighborhood information more adaptable. In each layer, the hidden states (or features at the beginning) are modified by a matrix (a linear layer). The final hidden states of a node are determined by averaging its neighbors' hidden states using learned weights.
\end{description}

In this evaluation, we refer to a model with $k$ layers by appending the layers to the model's abbreviated name, e.g., a three-layer GraphSage model is referred to as Sage-3.

\paragraph{Datasets}

We use the datasets shown in Table~\ref{tab:datasets}, representing a diverse range of graph sizes and characteristics.

\begin{table}[h]
 \centering
 \caption{Datasets used in the experiments.}
 \label{tab:datasets}
 \setlength{\tabcolsep}{0.3em}
 \small{
 \begin{tabular}{|l|l|r|r|r|}
 \hline
 \textbf{Name} & \textbf{Dataset} & \textbf{Vertices} & \textbf{Edges} & \textbf{Features} \\
 \hline
 \textbf{OGBN-Ar} & OGBN ArXiv~\cite{ogb} & 0.16M & 1.11M & 128\\
 \textbf{Reddit} & Reddit~\cite{DBLP:conf/nips/HamiltonYL17} & 0.22M & 109.30M & 602 \\
 \textbf{OGBN-Pr} & OGBN Products~\cite{ogb} & 2.34M & 58.99M & 100 \\
 \textbf{OGBN-Pa} & OGBN Papers100M~\cite{ogb} & 105.92M & 1.51B & 128 \\
 \textbf{RMAT-X} & RMAT Scale X~\cite{rmat} & $2^X$ & $2^{X+4}$ & 128 \\
 \hline
 \end{tabular}
 }
\end{table}

In addition to these real-world datasets, we use synthetic graphs based on the Graph500 RMAT generator~\cite{rmat} in scalability experiments. RMAT graphs support different scales, denoted RMAT-X, corresponding to a synthetic graph with $2^{X}$ nodes following a power-law degree distribution with an average degree of 16. We use a feature dimension of 128 for all RMAT graphs. RMAT-26 contains approximately $2^{30} \approx 1.07 \ \text{billion}$ edges; RMAT-30 contains approximately $2^{34} \approx 17.18 \ \text{billion}$; RMAT-36 contains approximately $2^{40} \approx 1.10 \ \text{trillion}$. We use the official graph splits where available.

\paragraph{Baselines}

We compare \systemname{} against three carefully selected baselines: DGL~\cite{dgl}, MGG~\cite{mgg}, and Cluster-GCN~\cite{cluster-gcn}. These baselines represent a broad spectrum of GNN systems, covering both distributed and single-node settings and different training paradigms (full-graph and sampling).

\begin{description}
 \item[DGL] is a widely used distributed GNN training framework that supports both full-graph and sampling-based training. It is the closest state-of-the-art system to \systemname{}, as both systems use CPUs for sampling and GPUs for forward and backward computation. DGL provides a direct comparison with an established distributed GNN system.
 \item[MGG] is a recent, high-performance full-graph system optimized for multi-GPU training on a single node. It uses fast NVLink bandwidth between GPUs but does not support distributed execution. We include MGG as it represents the best single-node full-graph training system, providing a useful contrast to \systemname{}'s distributed design.
 \item[Cluster-GCN] precomputes clustered training batches to reduce runtime sampling and communication. It works only for GCN models, highlighting the specialization–generality trade-off in GNN systems.
\end{description}

We do not include ByteGNN~\cite{DBLP:journals/pvldb/ZhengCCSWLCYZ22}, PaGraph~\cite{DBLP:conf/cloud/LinLMLX20}, or MariusGNN~\cite{DBLP:conf/eurosys/WaleffeMRV23} quantitatively because they target different system points (CPU-centric or out-of-core training) and execution models; a direct numerical comparison would be apples-to-oranges and would obscure the distributed GPU setting we study. We therefore discuss them qualitatively in Section~\ref{sec:related-work}.

\paragraph{Hardware and Configuration}

We use two different systems for the evaluation. We evaluate distributed training systems on a cluster of 8 identical machines. Each machine is equipped with 2 Intel CPUs (Xeon Gold 6134M), 764~GB of DDR4 memory, and two Nvidia Tesla V100 PCIe 32~GB GPUs. These machines are connected by 2$\times$ 10~Gb/s network (joined as one bonding interface). For single node evaluation, we also use a machine with 2 AMD CPUs (EPYC 9654), 2~TB of DDR5, and 8 Nvidia Hopper H100 80~GB SXM GPUs that are fully connected with 900~GB/s NVLink. The operating system is a 64-bit Debian Linux (kernel version: 5.15); we use CUDA 11.8 and run all experiments in docker containers. We include results on the H100 NVLink node (e.g., Table~\ref{tab:mgg}) to reflect high-bandwidth intra-node interconnects.

\paragraph{Timing accounting} Unless stated otherwise, reported \emph{epoch time} excludes repartitioning/switching time; switching happens between super-epochs, is overlapped with I/O, and is reported separately in Table~\ref{tab:overheads}. We report the mean of 3 runs.

\paragraph{Comparability notes} Our V100 cluster experiments (Tables~\ref{tab:overall-comparison}--\ref{tab:ablation}) compare distributed training systems under realistic multi-node network constraints. H100 NVLink results (Table~\ref{tab:mgg}) evaluate single-node full-graph systems under best-case interconnect bandwidth, showing that \systemname{} remains competitive even when baselines enjoy NVLink speeds. These two settings answer different questions: the cluster shows scalability across nodes; the single node shows that isolation does not sacrifice per-GPU efficiency.

We strive to ensure comparable and fair results across all systems in each experiment. We maintain consistency by using the same batch size (1000) for all mini-batch training systems. We set the number of sampled neighbors at each layer to \{25,10\}, \{15,10,5\}, and \{20,15,10,5\} for models with 2, 3, and 4 layers, respectively. Moreover, we standardize key hyperparameters: a hidden dimension of 128, a learning rate of 0.003, a dropout rate of 0.5, and a training duration of 500 epochs are applied uniformly. We use the batch-level version of \systemname{}'s gradient correction $\nabla_\theta \tilde{L}^{\text{batch}}(c)$ \eqref{eq:batch-estimator} with resampling-based correction factors $c^t_\text{resampling}$ \eqref{eq:resampling}. While the batch-level correction is an approximation of the node-level estimator, it preserves convergence behavior (Section~\ref{sec:convergence}) and achieves comparable or better accuracy than the minimally-biased variant (Section~\ref{sec:minbiased}). We also conducted hyperparameter sensitivity experiments that showed almost no discernible differences in performance across different configurations when comparing the systems, indicating that our chosen settings are representative and fair. Finally, we make a concerted effort to optimally configure all systems under comparison and run each experiment multiple times to ensure statistically significant results.

\subsection{Training Efficiency and Accuracy (RQ1)}
\label{subsec:eval-overall-comparison}

\subsubsection{Epoch Time}

We compare epoch time between \systemname{} and DGL for datasets in Table~\ref{tab:datasets} and 2--4 layer models, using 16 partitions on 16 GPUs (Table~\ref{tab:overall-comparison}).

\begin{table*}[t]
   \centering
   \caption{Average epoch time and test accuracy across all systems, models, model configurations, and datasets. Datasets are partitioned into 16 partitions and run in 16 GPUs. We run this experiment on our 8-machine V100 cluster.
}
   \label{tab:overall-comparison}
   \setlength{\tabcolsep}{0.5em}
\begin{tabular}{|l|l|r|r|r|r||r|r|r|r|r|}
\hline
        &            & \multicolumn{4}{c||}{\textbf{Average epoch time}} & \multicolumn{4}{c|}{\textbf{Test accuracy}} \\
\hline
\textbf{Model}& \textbf{Systems}&\textbf{OGBN-Ar} &\textbf{Reddit}&\textbf{OGBN-Pr}&\textbf{OGBN-Pa}&\textbf{OGBN-Ar}&\textbf{Reddit}&\textbf{OGBN-Pr}&\textbf{OGBN-Pa} \\
\hline
 \textbf{Sage-2} & DGL-Random  & 0.6 s  & 6.2 s    & 4.2 s     & 15.5 s  & 0.5630 & 0.9626   & 0.7748   & 0.4877 \\
 \textbf{Sage-2} & DGL-Metis   & 0.4 s  & 3.7 s    & 2.3 s     & 9.1 s   & 0.5649 & 0.9633   & 0.7756   & 0.4866\\
 \textbf{Sage-2} & \systemname  & 0.4 s  & 1.0 s    & 1.3 s     & 3.9 s   & 0.5571 & 0.9633   & 0.7688   & 0.4676 \\

\hline
\textbf{GCN-2} & DGL-Random    & 0.6 s  & 6.2 s    & 4.2 s     & 15.0 s  & 0.5078 & 0.9211   & 0.7651   & 0.4734  \\
\textbf{GCN-2} & DGL-Metis     & 0.4 s  & 3.7 s    & 2.3 s     & 9.1 s   & 0.5055 & 0.9209   & 0.7676   & 0.4761  \\
\textbf{GCN-2} & \systemname    & 0.4 s  & 1.0 s    & 1.3 s     & 3.8 s   & 0.5416 & 0.9472   & 0.7663   & 0.4886  \\

\hline
\textbf{GAT-2} & DGL-Random    & 0.7 s  & 6.3 s    & 4.4 s     & 14.9 s  & 0.5682 & 0.9493   & 0.7881   & 0.5115 \\
\textbf{GAT-2} & DGL-Metis     & 0.5 s  & 3.8 s    & 2.4 s     & 9.9 s   & 0.5741 & 0.9525   & 0.7902   & 0.5130  \\
\textbf{GAT-2} & \systemname    & 0.4 s  & 1.1 s    & 1.4 s     & 5.1 s   & 0.5430 & 0.9476   & 0.7829   & 0.4984  \\

\hline
\textbf{Sage-3} & DGL-Random   & 0.9 s  & 9.7 s    & 8.5 s     & 22.9 s  & 0.5553 & 0.9593   & 0.7513   & 0.4435  \\
\textbf{Sage-3} & DGL-Metis    & 0.6 s  & 6.2 s    & 4.5 s     & 16.1 s  & 0.5567 & 0.9635   & 0.7527   & 0.4489  \\
\textbf{Sage-3} & \systemname   & 0.4 s  & 1.0 s    & 1.4 s     & 4.9 s   & 0.5532 & 0.9589   & 0.7738   & 0.4259  \\

\hline
\textbf{GCN-3} & DGL-Random    & 0.9 s  & 9.8 s    & 8.5 s     & 23.2 s  & 0.4984 & 0.6527   & 0.6132   & 0.4355  \\
\textbf{GCN-3} & DGL-Metis     & 0.6 s  & 6.0 s    & 4.5 s     & 16.1 s  & 0.4857 & 0.6593   & 0.6164   & 0.4384  \\
\textbf{GCN-3} & \systemname    & 0.4 s  & 1.0 s    & 1.4 s     & 4.8 s   & 0.5172 & 0.9326   & 0.7502   & 0.4333  \\

\hline
\textbf{GAT-3} & DGL-Random    & 1.0 s  & 10.0 s   & 9.0 s     & 23.7 s  & 0.5688 & 0.8999   & 0.7160   & 0.4971  \\
\textbf{GAT-3} & DGL-Metis     & 0.7 s  & 6.4 s    & 5.1 s     & 20.4 s  & 0.5755 & 0.8650   & 0.7182   & 0.5130 \\
\textbf{GAT-3} & \systemname    & 0.5 s  & 1.1 s    & 1.6 s     & 5.8 s   & 0.5487 & 0.9251   & 0.7855   & 0.4902  \\

\hline
\textbf{Sage-4} & DGL-Random   & 1.5 s  & 16.6 s   & 40.2 s    & 54.2 s  & 0.5237 & 0.9514   & 0.6692   & 0.4237  \\
\textbf{Sage-4} & DGL-Metis    & 1.0 s  & 13.1 s   & 20.0 s    & 39.2 s  & 0.5285 & 0.9486   & 0.6718   & 0.4297  \\
\textbf{Sage-4} & \systemname   & 0.4 s  & 1.3 s    & 1.5 s     & 5.7 s   & 0.5428 & 0.9633   & 0.7946   & 0.4174  \\

\hline
\textbf{GCN-4} & DGL-Random    & 1.5 s  & 16.7 s   & 40.1 s    & 53.8 s  & 0.5121 & 0.2540   & 0.2658   & 0.4113  \\
\textbf{GCN-4} & DGL-Metis     & 1.0 s  & 13.1 s   & 20.0 s    & 39.4 s  & 0.5203 & 0.1091   & 0.2737   & 0.4204  \\
\textbf{GCN-4} & \systemname    & 0.4 s  & 1.3 s    & 2.5 s     & 5.7 s   & 0.5190 & 0.9180   & 0.7124   & 0.4203  \\

\hline
\textbf{GAT-4} & DGL-Random    & 1.7 s  & 18.0 s   & 55.6 s    & 69.0 s  & 0.5619 & 0.6755   & 0.4423   & 0.4833  \\
\textbf{GAT-4} & DGL-Metis     & 1.1 s  & 14.0 s   & 22.3 s    & 43.9 s  & 0.5735 & 0.6476   & 0.4430   & 0.4876  \\
\textbf{GAT-4} & \systemname    & 0.5 s  & 1.5 s    & 3.2 s     & 7.6 s   & 0.5249 & 0.9141   & 0.7784   & 0.4914  \\
\hline

\end{tabular}

\end{table*}

DGL-Metis is faster than DGL-Random via fewer edge cuts, but adds expensive preprocessing. \systemname{} avoids cross-partition feature/activation exchange within iterations (gradient-only), yielding $1\times$--$13\times$ speedups ($4\times$ average). In summary, \ul{\systemname{} cuts training time by up to an order of magnitude compared to DGL.}

\subsubsection{Accuracy}

Having established \systemname{}'s speed advantage, we now examine model quality. Table~\ref{tab:accuracy} summarizes accuracy. \systemname{} averages 68\% across the tested scenarios, compared to 61\% for both DGL systems. This gap grows with depth: 2-layer models are comparable (about 69\%), while DGL drops to around 63\% and 51\% for 3- and 4-layer models, respectively, whereas \systemname{} stays near 67\%. We hypothesize that this may stem from a dropout-like effect in \systemname{}'s isolated training, where input nodes are effectively removed at random, potentially reducing overfitting; the ablation study in Section~\ref{subsec:eval-overheads-ablation} provides supporting evidence. Convergence behavior is similar to DGL and discussed in Section~\ref{sec:convergence}.
\begin{table}[H]
\centering
 \caption{Average test accuracy across model depths.}
 \label{tab:accuracy}
\begin{tabular}{|l|ccc|c|}
\hline
& \textbf{2-layer} & \textbf{3-layer} & \textbf{4-layer} & \textbf{Overall} \\
\hline
\textbf{DGL-Random} & 0.6894 & 0.6326 & 0.5145 & 0.6122 \\
\textbf{DGL-Metis} & 0.6909 & 0.6328 & 0.5045 & 0.6094 \\
\textbf{\systemname{}} & 0.6894 & 0.6746 & 0.6664 & 0.6768 \\
\hline
\end{tabular}
\end{table}
Table~\ref{tab:accuracy} shows 2-layer models are comparable across systems, while 3- and 4-layer models see large gaps in favor of \systemname{}, consistent with deeper models suffering more from cross-partition sampling noise. In summary, \ul{\systemname{} achieves comparable or better training accuracy than other state-of-the-art GNN training systems, with the gap widening as models get deeper.}

\subsubsection{Sampling Overhead}

We compare \systemname{} with Cluster-GCN, a sampling-based system that reduces communication via precomputed clusters. Cluster-GCN eliminates on-the-fly sampling and reduces cross-partition overhead. Table~\ref{tab:cluster-gcn} shows the average epoch time and test accuracy on GCN-2 for OGBN-Pr.

\begin{table}[h]
\centering
 \caption{Average epoch time and test accuracy of Cluster-GCN and \systemname{} on GCN-2 for OGBN-Pr. Experiment run on our H100 machine.}
 \label{tab:cluster-gcn}
\begin{tabular}{|l|c|c|}
\hline
\textbf{System} & \textbf{Epoch time} & \textbf{Test accuracy} \\ \hline
\textbf{Cluster-GCN} & 1.2 s & 0.7390 \\
\textbf{\systemname{}} & 2.6 s & 0.7663 \\ \hline
\end{tabular}
\end{table}

Although Cluster-GCN achieves an epoch time of around half that of \systemname{}, \systemname{} consistently achieves better model accuracy. Moreover, creating the clusters for OGBN-Pr takes 881 seconds for this dataset. In comparison, \systemname{}'s total overhead for this experiment amounts to 49 (partitioning) + 76 (all super-epoch switches) = 125 seconds. If we consider a generous 500-epoch training run, the total time for Cluster-GCN is 1481 seconds, whereas \systemname{} requires a total time of 1425. Finally, it is noteworthy that Cluster-GCN is designed explicitly for GCNs and does not extend its benefits to other GNN architectures. Cluster-GCN also eliminates most neighbor traffic via precomputed clusters whereas our approach eliminates it during iterations and retains generality across GNN architectures. Specialized methods like Cluster-GCN underscore the benefits of pre-processing to reduce communication overhead but cannot support broader applications.

\subsubsection{Full-Graph Training Time}

Finally, we compare \systemname{} with MGG, a high-performance single-node full-graph system using NVLink. Table~\ref{tab:mgg} shows epoch time for GCN-2 across datasets and GPU counts; \systemname{} runs in full-graph mode here.

\begin{table}[h]
\centering
 \caption{Average epoch time for MGG and \systemname{} (full-graph training mode) to train a GCN-2 model for different datasets and across different GPU configurations. Experiment run on a machine with 8$\times$H100 with full NVLink connections.}
 \label{tab:mgg}
\begin{tabular}{|l|l|c|c|c|}
\hline
\textbf{Dataset} & \textbf{Method} & \textbf{2 GPU} & \textbf{4 GPU} & \textbf{8 GPU} \\ \hline
\multirow{2}{*}{Reddit}
& \textbf{MGG} & 0.2 s & 0.2 s & 0.2 s \\ \cline{2-5}
& \textbf{\systemname{}} & 0.4 s & 0.2 s & 0.1 s \\ \hline
\multirow{2}{*}{OGBN-Pr}
& \textbf{MGG} & 0.3 s & 0.2 s & 0.2 s \\ \cline{2-5}
& \textbf{\systemname{}} & 0.4 s & 0.2 s & 0.1 s \\ \hline
\multirow{2}{*}{OGBN-Pa}
& \textbf{MGG} & Crash & Crash & 2.8 s \\ \cline{2-5}
& \textbf{\systemname{}} & 5.7 s & 3.5 s & 2.4 s \\ \hline
\end{tabular}
\end{table}

\begin{figure*}[t]
 \centering
 \includegraphics[width=0.23\textwidth]{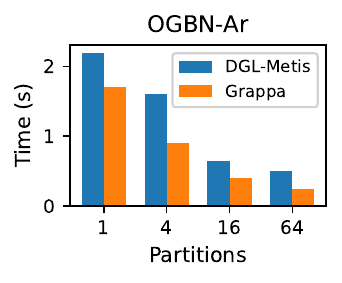}
 \hfill
 \includegraphics[width=0.23\textwidth]{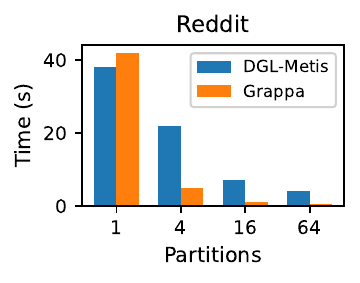}
 \hfill
 \includegraphics[width=0.23\textwidth]{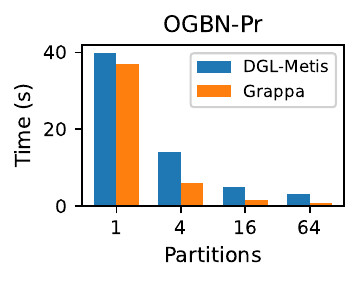}
 \hfill
 \includegraphics[width=0.23\textwidth]{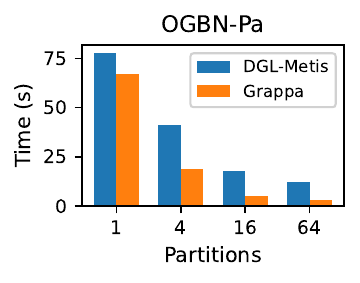}
 \caption{Epoch time vs.\ number of partitions for \systemname{} and DGL (Sage-3, 8$\times$V100 cluster).}
 \label{fig:scalability-partitions}
\end{figure*}

MGG achieves subsecond epoch time on small datasets (OGBN-Pr and Reddit) and 2.1 seconds for the larger OGBN-Pa, faster than similar systems such as BNS-GCN~\cite{DBLP:conf/mlsys/WanLLKL22} and NeutronStar~\cite{DBLP:conf/sigmod/WangZWCZY22}. However, MGG must store the entire graph on HBM and runs out of memory when processing OGBN-Pa with only 2 and 4 GPUs. In comparison, \systemname{} avoids this memory limitation and speeds up full-graph epoch times by eliminating communication overhead.

Even if MGG supported multi-node execution, it would need expensive communication links and would still require storing the entire graph on GPUs. This constraint shows why \systemname{}'s distributed architecture can handle larger datasets across multiple nodes without the same memory limitations. Key takeaway: \systemname{}'s flexible and isolated training design outperforms single-node systems like MGG without suffering from the same communication bottleneck.

\subsection{Scalability and Capacity (RQ2)}
\label{subsec:eval-scalability}

We now evaluate the scalability of \systemname{} as we increase the number of partitions and the graph size. We also show how \systemname{}'s phase-parallel training allows the system to train very large graphs using limited resources.

\subsubsection{Increasing Number of Partitions}

We measure the average epoch time using Sage-3 and the different datasets in Table~\ref{tab:datasets} as we increase the number of partitions from 1 to 64. Figure~\ref{fig:scalability-partitions} shows the training time per epoch for \systemname{} and the DGL baseline with METIS partitioning as we vary the number of partitions.

DGL scales poorly compared to \systemname{} as partitions increase, despite high-quality METIS cuts. For instance, from 1 to 16 partitions, DGL achieves only $\sim 3$--$5 \times$ speedup instead of the ideal $16 \times$, due to growing sampling communication overheads. In contrast, \systemname{} scales almost linearly ($\sim 16 \times$) on all datasets except OGBN-Ar, whose small, highly connected graph amplifies system overheads relative to epoch time. In summary, \ul{\systemname{}'s training time scales nearly linearly with partition count.}

\subsubsection{Increasing Graph Size}

Having shown that \systemname{} scales with partition count, we now fix the partition count and increase graph size. To that end, we train Sage-3 on RMAT~\cite{rmat} synthetic graphs and generate random vertex features, doubling the graph size from RMAT-26 to RMAT-30 each time. By sampling each partition in isolation, \systemname{} avoids superlinear blowups from communication: epoch time grows from 500\,s (RMAT-26) to 11,000\,s (RMAT-30), scaling roughly with the $16\times$ increase in edges. In summary, \ul{\systemname{}'s training time scales almost linearly with the graph size.}

\subsubsection{Training Very Large Graphs}

We now push capacity to the limit by training on graphs that do not fit on a single machine. We use Sage-3 and RMAT-36, partition it into 256 partitions, and train one partition at a time using our phase-parallel approach (Section~\ref{subsec:design-phase-training}). We find that \systemname{} can run one epoch of Sage-3 on RMAT-36 in 3.6 hours on a single V100 machine. While training such a model to converge would take a long time, these results show that it can be done on a single machine if needed. \ul{This result demonstrates the feasibility of processing trillion-edge topologies on a single machine in phase-parallel mode with gradient-only communication, under a realistic memory budget.}

\subsection{Ablation Study and Overheads (RQ3)}
\label{subsec:eval-overheads-ablation}

Finally, we conduct an ablation study to quantify the contribution of each component and the overheads of \systemname{}.

\subsubsection{Ablation Study}

We conduct an ablation study to assess the contribution of each component of \systemname{} by comparing it with three variants:
\emph{\systemname{}-UW}, without coverage-corrected aggregation;
\emph{\systemname{}-FP}, with fixed (non‑switchable) partitions; and
\emph{\systemname{}-50S}, which relaxes isolation by sampling 50\% of vertices from remote partitions. Table~\ref{tab:ablation} presents the results of this experiment.

\begin{table}[h]
 \centering
 \caption{Ablation study. We train each of our 3-layer models on the Reddit dataset with a fixed number of 16 partitions on our 8-machine V100 cluster for up to 500 epochs and report the test accuracy.}
 \label{tab:ablation}
 \setlength{\tabcolsep}{0.5em}
 \small{
 \begin{tabular}{|l|c|c|c|}
 \hline
 & \textbf{GCN-3} & \textbf{Sage-3} & \textbf{GAT-3} \\
 \hline
 \textbf{\systemname{}} & 0.9326 & 0.9589 & 0.9251 \\
 \hline
 \textbf{\systemname{}-UW} & 0.8284 & 0.9636 & 0.8403 \\
 \textbf{\systemname{}-FP} & 0.8464 & 0.9556 & 0.9068 \\
 \textbf{\systemname{}-50S} & 0.6593 & 0.9635 & 0.8999 \\
 \hline
 \end{tabular}
 }
\end{table}

\emph{\systemname{}-FP} reduces accuracy across all models, confirming that repartitioning is essential. \emph{\systemname{}-UW} hurts GCN and GAT substantially; Sage is less sensitive and even slightly higher here, suggesting that coverage correction matters most for models that aggregate more aggressively. Most striking is \emph{\systemname{}-50S}: despite 50\% remote neighbor access, it underperforms full \systemname{} on GCN-3 (0.66 vs.\ 0.93) and GAT-3 (0.90 vs.\ 0.93). One possible explanation is a regularization effect analogous to dropout that outweighs missing-neighbor information loss. The effect is model-dependent: GAT and GCN benefit more from isolation than Sage, where \emph{\systemname{}-50S} edges ahead.
Each component of \systemname{} contributes to model quality, and full isolation paired with correction outperforms partial isolation with more neighbor access.

\subsubsection{Overheads}

\systemname{}'s design introduces three kinds of overheads: partitioning, switching partitions, and redundant storage. Table~\ref{tab:overheads} summarizes the first two.

\begin{table}[h]
 \centering
 \caption{Overhead analysis on 8$\times$V100 cluster. Partitioning: OGBN-Pa. Switching: Sage-3 on OGBN-Pr.}
 \label{tab:overheads}
 \setlength{\tabcolsep}{0.5em}
 \small{
 \begin{tabular}{|l|r|r|r|}
 \hline
 & \textbf{4 parts.} & \textbf{16 parts.} & \textbf{64 parts.} \\
 \hline
 \textbf{Partitioning (\systemname{})} & 2,110 s & 2,568 s & 2,345 s \\
 \textbf{Partitioning (DGL METIS)} & 6,262 s & 6,523 s & 7,292 s \\
 \hline
 \textbf{Switching time} & 128 s & 76 s & 51 s \\
 \textbf{Switching overhead} & 4\% & 10\% & 12\% \\
 \hline
 \end{tabular}
 }
\end{table}

Partitioning time remains stable as \systemname{} uses random graph partitioning, which is $2$--$5\times$ faster than DGL's METIS while still achieving high accuracy through repartitioning and correction. Switching overhead is at most 12\% of training time, which pales compared to the $>6\times$ speedups from isolated training.

\begin{table}[h]
 \centering
 \caption{Size of each partition and the total size of all partitions for different datasets using \systemname{}'s partitioning with 4 partitions. The numbers are in GB.}
 \label{tab:partition-storage}
 \setlength{\tabcolsep}{0.5em}
 \small{
 \begin{tabular}{|l|r|r|r|r|r|}
 \hline
 \textbf{Dataset} & \textbf{Part. 1} & \textbf{Part. 2} & \textbf{Part. 3} & \textbf{Part. 4} & \textbf{Total} \\
 \hline
 \multicolumn{6}{|c|}{\textbf{\systemname{}'s Random Partitioning}} \\
 \hline
 \textbf{OGBN-Ar} & 0.08 & 0.08 & 0.08 & 0.08 & 0.11 \\
 \textbf{Reddit} & 1.37 & 1.38 & 1.38 & 1.37 & 2.27 \\
 \textbf{OGBN-Pr} & 1.83 & 1.83 & 1.83 & 1.83 & 2.86 \\
 \textbf{OGBN-Pa} & 56.70 & 56.68 & 56.68 & 56.69 & 78.69 \\
 \hline
 \multicolumn{6}{|c|}{\textbf{DGL's METIS Partitioning}} \\
 \hline
 \textbf{OGBN-Ar} & 0.06 & 0.06 & 0.06 & 0.06 & 0.11 \\
 \textbf{Reddit} & 1.54 & 1.40 & 1.10 & 1.26 & 2.27 \\
 \textbf{OGBN-Pr} & 1.40 & 1.47 & 1.60 & 1.60 & 2.86 \\
 \textbf{OGBN-Pa} & 39.01 & 46.74 & 46.54 & 40.10 & 78.69 \\
 \hline
 \end{tabular}
 }
\end{table}

Table~\ref{tab:partition-storage} shows partition sizes: random partitioning yields nearly equal sizes, so speedups are not from load imbalance. METIS produces smaller but less balanced partitions; \systemname{} handles both.

\subsubsection{Minimally-Biased Estimator Results}
\label{sec:minbiased}

In the main text, we presented the results for the batch-level estimator $\nabla_\theta \tilde{L}^{\text{batch}}(c)$ \eqref{eq:batch-estimator} with resampling-based correction factors $c^t_\text{resampling}$ \eqref{eq:resampling}. Table~\ref{tab:minimally-biased} shows the test accuracy results for the minimally-biased estimator with equal selection probabilities for neighbors $c^t_\text{uniform}$~\eqref{eq:correction_uniform}. Comparing the accuracy to the results of the resampling-based estimator $c^t_\text{resampling}$~\eqref{eq:resampling} is inconclusive, but the resampling-based estimator has better accuracy in $10/18$ cases and in all cases for the GCN architecture and is more stable. More stability is expected for an estimator that achieves shrinkage and better accuracy hints at the benefits of accepting some degree of bias to avoid overfitting.
\begin{table}[h]
 \centering
 \caption{Test accuracy for the minimally-biased gradient estimator $\nabla_\theta \tilde{L}^{\text{batch}}(c)$ \eqref{eq:batch-estimator} with correction factors for a uniform neighbor distribution $c^t_\text{uniform}$~\eqref{eq:correction_uniform}.}
 \label{tab:minimally-biased}
 \setlength{\tabcolsep}{0.5em}
 \small{
 \begin{tabular}{|l|c|c|c|}
 \hline
 & \textbf{OGBN-Ar} & \textbf{Reddit} & \textbf{OGBN-Pr} \\
 \hline
 \textbf{Sage-2} & 0.5590 & 0.9627 & 0.7706 \\
 \textbf{GCN-2} & 0.4629 & 0.8943 & 0.6630 \\
 \textbf{GAT-2} & 0.5640 & 0.9483 & 0.7874 \\
 \hline
 \textbf{Sage-3} & 0.5529 & 0.9611 & 0.7815 \\
 \textbf{GCN-3} & 0.4433 & 0.8414 & 0.6919 \\
 \textbf{GAT-3} & 0.5714 & 0.9083 & 0.7786 \\
 \hline
 \end{tabular}
 }
\end{table}
\subsubsection{Convergence Behavior}\label{sec:convergence}

We evaluate the convergence of each system in a selected subset of scenarios from Table~\ref{tab:overall-comparison}. We pick GCN-2 as the model on the OGBN-Pr dataset. Figure~\ref{fig:training-convergence} shows the evolution of the training accuracy.

\begin{figure}[t]
 \centering
 \includegraphics[width=0.35\textwidth]{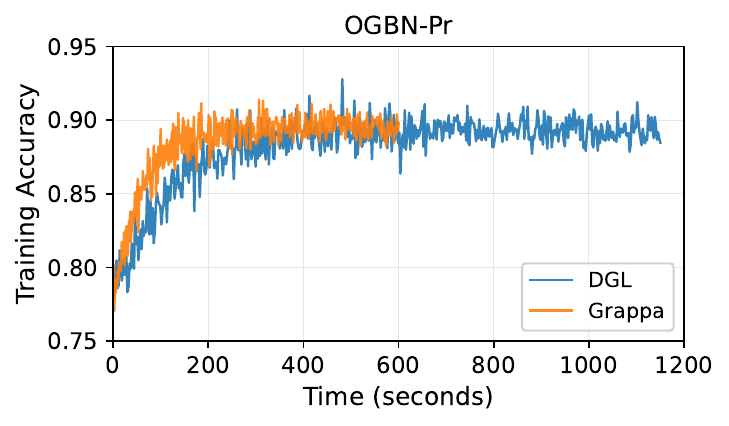}
 \caption{Training accuracy over time for GCN-2 with 16 partitions (500 epochs). Both systems converge similarly; \systemname{}'s epoch time is almost half of DGL's.}
 \label{fig:training-convergence}
\end{figure}

We observe no meaningful differences in training accuracy across systems. Moreover, in each case, the training process converges similarly and within roughly the same number of epochs. However, \systemname{}'s epoch time is almost half of DGL's. The other models and datasets in Table~\ref{tab:overall-comparison} exhibit similar patterns.

\section{Related Work}
\label{sec:related-work}

\paragraph{Distributed GNN Training} Systems such as P3~\cite{DBLP:conf/osdi/GandhiI21}, DistDGL~\cite{DBLP:conf/sc/Zheng0WZSSGZK20}, ROC~\cite{DBLP:conf/mlsys/JiaLGZA20}, Aligraph~\cite{aligraph}, BNS-GCN~\cite{DBLP:conf/mlsys/WanLLKL22}, NeutronStar~\cite{DBLP:conf/sigmod/WangZWCZY22}, AGL~\cite{DBLP:journals/pvldb/ZhangHL0HSGWZ020}, BGL~\cite{DBLP:conf/nsdi/LiuC00ZHPCCG23}, and ByteGNN~\cite{DBLP:journals/pvldb/ZhengCCSWLCYZ22} enable distributed training. P3 minimizes input feature transfer via model parallelism in the first layer, but subsequent hidden state exchange grows costly with wider dimensions~\cite{DBLP:conf/osdi/GandhiI21}. CPU-centric systems such as Dorylus~\cite{DBLP:conf/osdi/ThorpeQETHJWVNK21}, ByteGNN~\cite{DBLP:journals/pvldb/ZhengCCSWLCYZ22}, and AGL~\cite{DBLP:journals/pvldb/ZhangHL0HSGWZ020} suffer longer runtimes than GPU-based approaches. DistDGL~\cite{DBLP:conf/sc/Zheng0WZSSGZK20}, BGL~\cite{DBLP:conf/nsdi/LiuC00ZHPCCG23}, and AGL reduce communication through caching and partitioning. In contrast, \systemname{} does not rely on communication or caching optimizations as it removes cross-partition neighbor traffic entirely.

\paragraph{Multiple GPUs Training} Recent single‑node systems (MGG, Legion, XGNN, NeutronTask/NeutronTP) aggressively overlap or cache to tame PCIe/NVLink traffic~\cite{DBLP:conf/sosp/SongZ0C23,DBLP:conf/eurosys/YangTSWYCYZ22,DBLP:conf/usenix/SunSSSWWZLYZW23,DBLP:conf/cloud/LinLMLX20,mgg}. In contrast, \systemname{} does not rely on fast interconnects, multi‑GPU caches, and complex partitioning because it eliminates neighbor exchange altogether and caching overhead.

\paragraph{Graph Sampling Optimizations} GPU samplers like gSampler~\cite{DBLP:conf/sosp/GongLMC0LWL23} and NextDoor~\cite{DBLP:conf/eurosys/JangdaPGS21}, and CPU designs like FlashMob~\cite{DBLP:conf/sosp/YangMTCW21}, accelerate sampling but mostly in single‑machine settings. GraphSAINT~\cite{DBLP:conf/iclr/ZengZSKP20} improves scalability via subgraph sampling with unbiased normalization, extending GraphSAGE~\cite{DBLP:conf/nips/HamiltonYL17}. However, all still require remote neighbor access in partitioned graphs. \systemname{} removes neighbor exchange entirely, correcting coverage bias with importance weighting while remaining compatible with prior optimizations.

\paragraph{Precomputation / Out-of-Core} Cluster-GCN~\cite{cluster-gcn} uses METIS~\cite{metis} partitions for batching. \systemname{} instead forms partitions via a lightweight chunk-sweep schedule that covers cross-chunk edges across super-epochs. This sequential sweeping is simple, model-agnostic, prefetch-friendly, and avoids precomputation. Unlike systems that still exchange features or activations across partitions, \systemname{} communicates only gradients. MariusGNN~\cite{DBLP:conf/eurosys/WaleffeMRV23} targets single-machine out-of-core training, whereas \systemname{} scales capacity via distributed, phase-parallel execution without caching.

\paragraph{Local Averaging \& Federated GNNs} Our gradient‑only synchronization and phase‑parallel execution mirror local and decentralized SGD, which reduce communication through periodic averaging while still ensuring convergence under mild conditions~\cite{stich2018local,lian2017can}. Unlike federated GNNs, which exchange updates to address privacy and non‑IID client data~\cite{mcmahan2017communication,yao2023fedgcn,liu2024federated}, \systemname{} focuses on system scalability and bias correction through coverage‑aware scaling. Federated GNNs typically assume disjoint, non‑IID client subgraphs and prioritize privacy and partial participation; in that regime, client drift and structural heterogeneity dominate. \systemname{} operates on a centrally managed graph and can repartition to ensure long‑run neighbor coverage, so our coverage‑corrected scaling addresses system‑induced sampling bias rather than client drift. These techniques naturally complement federated training and may provide a path toward more efficient and robust federated GNNs.

\section{Conclusions}
\label{sec:conclusions}

\systemname{} demonstrates that the standard assumption—distributed GNN training requires cross-partition feature and activation exchange—can be relaxed. By training partitions in isolation with periodic repartitioning and coverage-corrected gradient aggregation, \systemname{} achieves up to $13 \times$ faster training while matching or exceeding state-of-the-art accuracy, particularly on deeper models where its dropout-like isolation reduces overfitting. Phase-parallel execution extends capacity to trillion-edge graphs on a single commodity server. Looking ahead, our coverage-correction and repartitioning techniques may transfer to federated and privacy-preserving GNN settings where cross-partition data exchange is restricted by design.

\bibliographystyle{ACM-Reference-Format}
\bibliography{bibliography/main}

\end{document}